\documentclass[11pt, letterpaper]{article}
\usepackage[utf8]{inputenc}
\usepackage[english]{babel}
\usepackage{amsfonts}
\usepackage{amssymb}
\usepackage{amsthm}
\usepackage{authblk}
\usepackage{pdfpages}
\usepackage[letterpaper, margin=1in]{geometry}

\usepackage{tikz}
\usetikzlibrary{calc}

\usepackage{float}
\usepackage[skip=1pt]{caption}

\title{Computational complexity lower bounds of certain discrete Radon transform approximations}
\author{Timur M. Khanipov}
\affil{\slshape Institute for Information Transmission Problems\\ \slshape of the Russian Academy of Sciences (IITP RAS)}
\date{January 3, 2018}

\newtheorem{theorem}{Theorem}[section]
\newtheorem{corollary}[theorem]{Corollary}
\newtheorem{proposition}[theorem]{Proposition}
\newtheorem{remark}[theorem]{Remark}
\newtheorem{definition}[theorem]{Definition}
\newtheorem{lemma}[theorem]{Lemma}

\newcommand{\patset}{\mathcal}

\begin{document}

\maketitle

\begin{abstract}
For the computational model where only additions are allowed, the $\Omega(n^2\log n)$ lower bound on operations count with respect to image size $n\times n$ is obtained for two types of the discrete Radon transform implementations: the fast Hough transform and a generic strip pattern class which includes the classical Hough transform, implying the fast Hough transform algorithm asymptotic optimality. The proofs are based on a specific result from the boolean circuits complexity theory and are generalized for the case of boolean $\vee$ binary operation. 
\end{abstract}


\section{Introduction}

The Radon transform of a function $f\colon\:\mathbb{R}^2\to\mathbb{R}$ maps all straight lines in $\mathbb{R}^2$ to line integrals of $f$:

\begin{equation}\label{radon}
l\mapsto\int\limits_l f ds
\end{equation}

\noindent It naturally arises in a large variety of applications including X-ray and other types of computed tomography, astronomy, electron microscopy, nuclear magnetic resonance, optics, stress analysis, geophysics and other areas~\cite{deans}. The Radon transform discrete analogues  are defined in at least two different ways. One approach is based on connection between the Radon and Fourier transforms~\cite{deans}, while the other one uses some sort of straightforward discretizations of equation~(\ref{radon}) and we will further call it ``discrete Radon transform''. A common variation of the discrete Radon transform is also known by the name ``Hough transform''~\cite{duda} when samples $f_{ij}$ are summed along all line patterns of type $j=[k i+b]$ or $i=[k j+b]$ for $(k,b)\in \patset{L}_n$ being a specific finite family of pattern parameters for a $n\times n$ image $I_n=\{f_{ij}\}$: 

\begin{equation}\label{hough}
(k,b)\mapsto \sum_{i=0}^{n-1} f_{i, [ki+b]},\qquad(k,b)\in\patset{L}_n.
\end{equation}

\noindent The Hough transform is well known in image processing mostly as means of detecting lines and is sometimes defined per se by~(\ref{hough}) without mentioning any ties to the Radon transform~\cite{shapiro}. If $|\patset{L}_n|$ grows as $n^2$ then a way of computing~(\ref{hough}) by direct summation would take $\Theta(n^3)$ operations which might be critically slow for some applications, especially real-time technical vision.

To mitigate this problem a different pattern family and the corresponding algorithm called the fast Hough transform were suggested, with each pattern being an approximation to~(\ref{hough}) lines and allowing simultaneous computation in~$\Theta(n^2 \log n)$ binary additions. The fast Hough transform was reinvented several times with the earliest known to the author published work being~\cite{gotz} in~1996. Brady's~1998 paper~\cite{brady} contains a thorough description of the algorithm and its properties including a logarithmic bound on pattern deviation from straight line. Due to efficient performance and possible in-place implementation, the fast Hough transform found many applications in various object recognition and technical vision areas: automatic documents orientation, vanishing points detection~\cite{FHT_underestimated} and even vehicle axles calculation in an intellectual transportation system~\cite{AVC} are only a few examples. It should be noted that the term ``fast Hough transform'' leads to some terminological confusion for it itself computes not the Hough transform~(\ref{hough}) but a specific approximation to it.

It was previously unknown whether the fast Hough transform could be computed in $o(n^2\log n)$ operations. Except for the trivial $\Omega(n^2)$ limitation no lower bounds were also known for the ``ordinary'' Hough transform~(\ref{hough}). In this paper we advance in these questions proving the $\Omega(n^2 \log n)$ lower bound for the minimal number of binary summations of both the fast Hough transform and a certain class of discrete Radon transform approximations covering~(\ref{hough}), when only addition operations are allowed. We use additive circuits as a natural computational model and exploit a specific theorem from the theory of boolean function complexity. In fact we prove a stronger result, also providing the $\Omega(n^2 \log n)$ bound for the case when pixels are considered to contain $\{0, 1\}$ values and summations are replaced with the boolean $\vee$ operation (using boolean instead of additive circuits).

The rest of the paper is organized as follows. Section~\ref{images} formally defines images, pixels and pattern sets and reproduces (from~\cite{CLB}) the necessary circuit complexity theory terms (adapted to images environment) as well as the key bounding theorem (without proof) which this paper is based on. In sections~\ref{FHT} and \ref{strip_section} the theorem is applied to the fast Hough transform and strip patterns resp. 
and  section~\ref{conclusion} summarizes the results and provides some directions for further research.

\section{Image patterns and complexity}\label{images}

We will consider raster images of height $h \geqslant 1$ and width $w \geqslant 2$ as a set $I$ of $w\cdot h$ variables (further called {\slshape pixels}) $p_{ij},\;i=0,...,h-1,\;j=0,...,w-1$ with values from either $\mathbb{F}_2=\{0, 1\}$ or $\mathbb{N} = \{0, 1, 2, ...\}$. We assume that $p_{00}$ resides in the bottom left corner. Where it does not lead to confusion the word {\slshape pixel} means both the variable and the value it might contain.

Suppose that a particular commutative semigroup operation is defined on pixel values: logical $\vee$ for $\mathbb{F}_2$ and arithmetic $+$ for $\mathbb{N}$. The words ``sum'', ``summation'' and the  $\Sigma$ symbol are freely used for both cases, this does not produce ambiguity because $\mathbb{F}_2$ and $\mathbb{N}$ are considered separately.

A {\slshape pattern} $T$ is a non-empty subset of $I$. A {\slshape pattern set $\patset{T}=\{T_k\}_{k=1}^{m}$} is a non-empty set of $m$ distinct patterns. Consider the task of simultaneously computing the following $m$ sums using only the semigroup operation ($+$ or $\vee$), generalizing~(\ref{hough}):

\begin{equation}\label{patterns_sum}
y_k = \sum_{p_\in T_k} p, \qquad\qquad k=1,...,m.
\end{equation}

An obvious way would be to perform a total of $\sum_{k=1}^m (|T_k|-1)$ binary operations, calculating each $y_k$ separately. The trick, however, is to exploit the patterns specific structure which might allow to greatly reduce the number of additions by reusing the sums of pixel subsets which reoccur in many patterns. A natural model for defining the complexity of problem~(\ref{patterns_sum}), which takes into account only the internal redundancy of patterns, is the circuit model. Further in this section we briefly describe it.

Let {\slshape circuit} be a directed acyclic graph with $w\cdot h$ input nodes  ({\slshape inputs}) $p_{ij}$ of zero fanin (no incoming edges) and $m$ output nodes ({\slshape outputs}) $y_k$ of zero fanout (no outgoing edges). Non-zero fanin nodes are called {\slshape gates} and may have any positive number of incoming edges, each gate performs summation over its fanin nodes. The {\slshape size} of a circuit is the total number of edges in it. We say that a circuit {\slshape computes} pattern set $\patset{T}$ if after performing all summations starting from the inputs, the outputs will contain the $y_k$ values according to~(\ref{patterns_sum}). At least one trivial computing circuit exists for any $\patset{T}$: take $m$ outputs each having $|T_k|$ incoming edges from the appropriate inputs, so the following definitions are correct:

\begin{definition}
OR-complexity $\mathrm{OR}(\patset{T})$ is the minimal size among all circuits computing $\patset{T}$ over the $(\mathbb{F}_2, \vee)$ semigroup. 
\end{definition}

\begin{definition}
SUM-complexity $\mathrm{SUM}(\patset{T})$ is the minimal size among all circuits computing $\patset{T}$ over the $(\mathbb{N}, +)$ semigroup. \end{definition}

For practical applications it might be more convenient to consider circuits containing only fanin-2 (i.e. with two or less incoming edges) nodes and count the number of gates instead of edges thus giving the complexity in terms of the minimal number of binary operations used. We will denote these measures as $\mathrm{OR2}(\patset{T})$ and $\mathrm{SUM2}(\patset{T})$ for $(\mathbb{F}_2, \vee)$ and $(\mathbb{N}, +)$ resp. From the following proposition, any OR/SUM lower bound yields almost the same OR2/SUM2 bound, so the latter practically more common case is reduced to the former one when considering lower bounds.

\begin{proposition}\label{sum2_proposition}For $L\in\{\mathrm{OR}, \mathrm{SUM}\}$ and any pattern set $\patset{T}$,
$L2(\patset{T}) \geqslant \frac12 L(\patset{T})$.
\end{proposition}
\begin{proof}
Indeed, for the number of gates $g$ and the number of edges $e$ in any fanin-2 circuit $G$ holds $e\leqslant 2g$ (adding a new gate produces at most two new edges). So if $G$ is the fanin-2 circuit with the minimal gates count, $\mathrm{L2}(\patset{T})=g\geqslant\frac12e\geqslant\frac12\mathrm{L}(\patset{T})$.
\end{proof}

\begin{remark}
The question of mere existence of a fanin-2 addition circuit computing $\patset{T}$ over $(\mathbb{N}, +)$ with a given number of binary gates can straightforwardly be reformulated as the so called ensemble computation problem and is thus NP-complete \cite{GARY}.
\end{remark}

It is easy to see that any circuit computing $\patset{T}$ over $(\mathbb{N}, +)$ would also compute $\patset{T}$ over  $(\mathbb{F}_2, \vee)$ after simply changing the summation operation (see~\cite[Introduction]{CLB} for a detailed explanation) which implies

\begin{proposition}\label{sum_or_ineq}
$\mathrm{SUM}(\patset{T}) \geqslant \mathrm{OR}(\patset{T})$ and$\;\;\mathrm{SUM2}(\patset{T}) \geqslant \mathrm{OR2}(\patset{T})$ for any $\patset{T}$.\qed
\end{proposition}

The inverse is not true and the complexity gap can be quite large (see, for example,~\cite[paragraph 5.1]{CLB}). In image processing the OR-case would seem somewhat ``exotic'' with the SUM-case being by far more common but due to proposition~\ref{sum_or_ineq} any lower bound on $\mathrm{OR}(\patset{T})$ provides the same lower bound on $\mathrm{SUM}(\patset{T})$, so it is sufficient to estimate only the stronger OR-case.

The two following simple propositions state that adding patterns to a set or extending it onto a larger image does not reduce its complexity. They are useful in further sections when we are reviewing only a distinguished subset (e.g. mostly vertical inclined to the right) of possible line patterns or consider subimages. 

\begin{proposition}\label{subpattern_prop}For $L\in\{\mathrm{OR}, \mathrm{SUM}, \mathrm{OR2}, \mathrm{SUM2}\}$ and pattern sets $\patset{M}\subseteq\patset{T}$,
$L(\patset{T}) \geqslant L(\patset{M})$.
\end{proposition}
\begin{proof}
Any circuit computing $\patset{T}$ can be reduced to the one computing $\patset{M}$ by removing the gates and edges which participate only in $\patset{T} \setminus \patset{M}$ computation.
\end{proof}

\begin{proposition}\label{subimage_prop}If $\patset{T}$ is a pattern set on image $I$, $I_0\subseteq I$ is a subimage and $\patset{M}=\{T\cap I_0 \mid T\in\patset{T} \}$ then for $L\in\{\mathrm{OR}, \mathrm{SUM}, \mathrm{OR2}, \mathrm{SUM2}\}$ holds $L(\patset{M})\leqslant L(\patset{T})$.
\end{proposition}
\begin{proof}
Indeed, any circuit computing $\patset{T}$ can be reduced to a circuit computing $\patset{M}$ by removing all nodes (and the corresponding edges) which depend only on pixels from $I\setminus I_0$ and possibly cleaning out the reoccurring outputs might they appear (this procedure is equivalent to permanently zeroing all these inputs).
\end{proof}

\noindent We can now define a few terms which are needed to formulate the lower bound theorem below.

\begin{definition}
{\slshape Volume} of a pattern set $\patset{T}$ 
$$
V(\patset{T}) = \sum_{T\in\patset{T}} |T|.
$$

\end{definition}

\begin{definition}
{\slshape Intersection area} of a pattern set $\patset{T}$ 
$$
S(\patset{T}) = |\patset{T}|\cdot |\bigcap\limits_{T\in\patset{T}}T|.
$$
\end{definition}

\begin{definition}
Self-intersection measure of a pattern set $\patset{T}$
$$
r(\patset{T}) = \max_{\varnothing\ne\patset{R} \subseteq \patset{T}} S(\patset{R}).
$$
\end{definition}

\begin{remark}
Terms $S(\patset{T})$, $V(\patset{T})$ and $r(\patset{T})$ directly correspond to terms ``area'', $|A|$ and $r(A)$ resp. from~\cite{CLB} for a boolean matrix $A$ corresponding to system~(\ref{patterns_sum}) when formally written in vector form $\mathbf{y} = A\mathbf{p}$ for a size $w\cdot h$ vector $\mathbf{p}=(p_{00}, p_{01}, ..., p_{h-1\:w-1})^T$ and a size $m$ vector $\mathbf{y}=(y_1,...,y_m)^T$. Patterns themselves correspond to the rows of $A$ in this representation. See~\cite{CLB} for more details.
\end{remark}

With this remark in mind we can now reformulate the area bounding theorem 3.12 from \cite[paragraph 3.4]{CLB} in the following way:

\begin{theorem}\label{bound_theorem}
For any pattern set $\patset{T}$ 

$$
  \mathrm{OR}(\patset{T}) \geqslant \frac{3 V(\patset{T})}{ r(\patset{T})}\log_3\frac{V(\patset{T})}{|\patset{T}|}.\qed
$$
\end{theorem}

Values $V(\patset{T})$ and $|\patset{T}|$ are easy to calculate but assessing $r(\patset{T})$ may be a non-trivial task. The next two sections give linear relatively to image size upper bounds of this term for the fast Hough transform and strip pattern sets which yield the desired lower bounds on complexity.

\section{Fast Hough transform}\label{FHT}
The fast Hough transform uses a special set of inductively defined image patterns constructed to approximate lines and ease their simultaneous computation. These patterns are defined separately for four symmetrical cases of mostly vertical or mostly horizontal, inclined to the right or left lines and it is enough to analyze only one of them. Simply described, FHT patterns $\patset{H}_0$ of order zero are single bottom pixels and patterns $\patset{H}_k$ of order $k>0$ are produced by putting patterns $\patset{H}_{k-1}$ of order $k-1$ onto their own tops with a possible additional horizontal shift of one pixel and cyclic wrapping over the right image border. Each pattern has width $1$ and height $2^k$. Distinct FHT pattern shapes generation for the $4\times 4$ case is shown in fig.~\ref{fht_shapes} with ``parent'' subshapes shaded at each stage. Every 4-pixel shape may start in 4 different bottom positions giving a total of 16 patterns (see fig.~\ref{fht_shifts} with all 4 positions of the shape in fig.~\ref{fht_shapes} below). 

For the formal definition we first assume that image height $h = 2^d,\;d\in\mathbb{N}$, this restriction is removed in the end of the section. We will need an auxiliary {\slshape horizontal span} function  $\Delta(T) = (j_{top} - j_{bot})\;\mathit{mod}\;w$, whith $j_{top}$ and $j_{bot}$ being horizontal coordinates of pattern $T$ top and bottom pixels resp. and a ``width-cyclic translation'' operator $\mathit{tran}_{a,b}(T) = \{p_{i+a,\:(j+b)\,\mathit{mod}\,w} \mid p_{ij}\in T\}.$

Now we can define mostly vertical inclined to the right {\slshape FHT patterns} $\patset{H}_{h,w}$ {\slshape of height} $h$, shortly also denoted as $\patset{H}_d$ {\slshape(``of order'' $d$)} for fixed $w$ (remember that $w\geqslant 2$), by induction:

\begin{enumerate}
\item $\patset{H}_0 = \{\{p_{00}\}, \{p_{01}\}, ..., \{p_{0\;{w-1}}\}\},$ 
\item $\patset{H}_{k} = \{T\:\cup\: \mathit{tran}_{\,2^{k-1},\,\Delta(T)+s}(T) \mid T\in\patset{H}_{k-1}, s\in\{0, 1\}\}$ for $k=1, 2, ..., d.$
\end{enumerate}

\begin{figure}
\begin{center}
\includegraphics[scale=0.75]{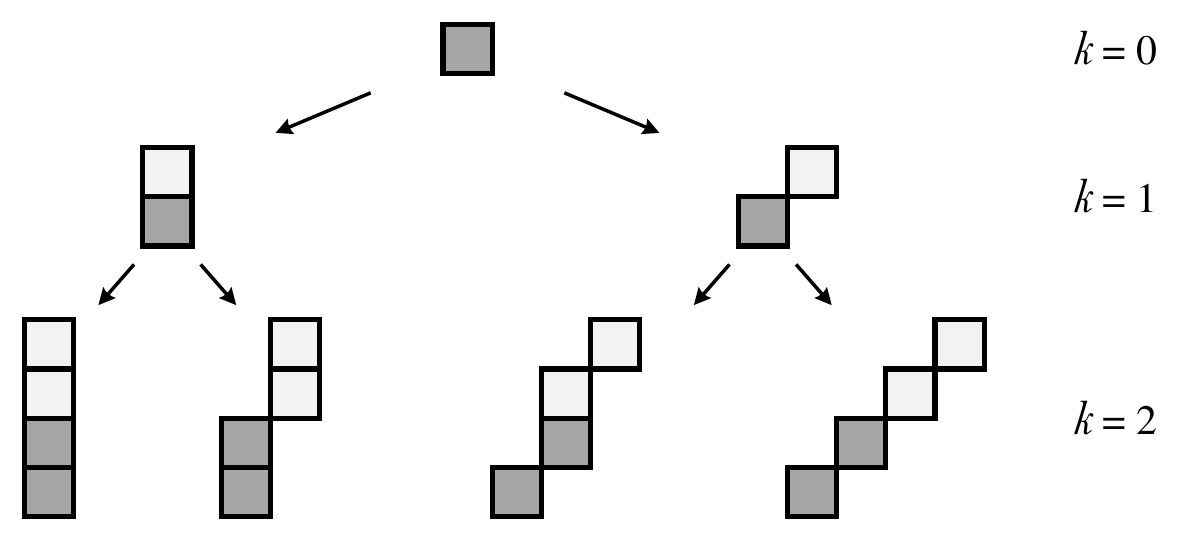}
\caption{Distinct 4-pixel FHT pattern shapes generation for $h=w=4$ with $\Delta=0, 1, 2, 3$ (from left to right below).}\label{fht_shapes}
\end{center}
\end{figure}

\begin{figure}
\begin{center}
\includegraphics[scale=0.75]{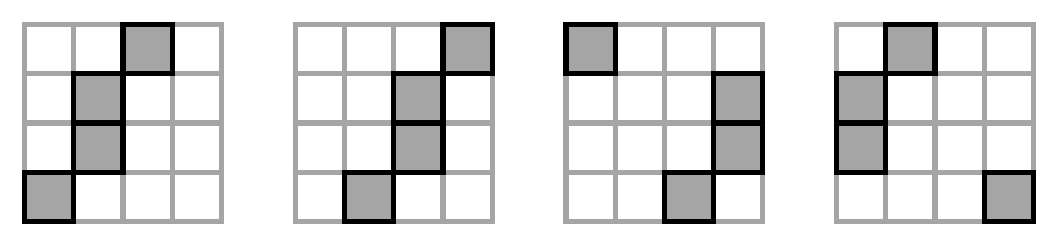}
\caption{Similar FHT patterns starting at different bottom pixels, $\Delta=2$.}\label{fht_shifts}
\end{center}
\end{figure}

The following propositions give some description to these patterns and follow directly from the definition or may easily be proved using induction (condition $w\geqslant 2$ is important). For convenience we use terms {\slshape bottom image (pixels)} $I^h_b = \{p_{ij} \mid 0\leqslant j<h/2\}$ and {\slshape top image (pixels)} $I^h_t = \{p_{ij} \mid h/2\leqslant j < h\}$.

\begin{proposition}\label{restriction}
For $h>1$ the restriction of $\patset{H}_{h,w}$ to top image is also FHT pattern of height $h/2$ (after appropriate pixels renumbering). The same obviously holds for the bottom image.\qed
\end{proposition}

\begin{proposition}\label{fht_size_proposition}
$|\patset{H}_{h,w}| = w\cdot h$.\qed
\end{proposition}

\begin{proposition}\label{fht_volume_proposition}
$\forall\,T\in\patset{H}_{h,w}\:\;|T| = h$ and, hence, $V(\patset{H}_{h,w}) = w\cdot h^2$.\qed
\end{proposition}

\begin{proposition}\label{sibling}
For every $T\in\patset{H}_{h,w}$ there is a unique $T'\in\patset{H}_{h,w}$ so that $T\cap I^h_b = T\cap T'$, such $T'$ is called a bottom-sibling of $T$. The same statement holds for $I^h_t$ (with top-siblings).\qed
\end{proposition}

\noindent For example, if the two leftmost patterns from fig.~\ref{fht_shapes} started from the same bottom pixel, they would be bottom-siblings.

\begin{lemma}\label{fht_lemma}
$r(\patset{H}_{h,w}) = h$.
\end{lemma}
\begin{proof}
$r(\patset{H}_{h,w}) \geqslant h$ because by proposition~\ref{fht_volume_proposition} for any $T\in\patset{H}_{h,w}$ intersection area $S(\{T\}) = h$. So it is enough to prove that $r(\patset{H}_{h,w}) \leqslant h$. 

We will use induction over $k$ from the definition of $\patset{H}_{h,w}$. The induction base is obvious:  $r(\patset{H}_0) = 1$. Supposing now that $r(\patset{H}_k) \leqslant 2^k$ we will show that $r(\patset{H}_{k+1}) \leqslant 2^{k+1}$ thus proving the induction step.

Indeed, let $\patset{M}\subseteq\patset{H}_{k+1}$ be a set of $q$ patterns with common pixels $P = \bigcap\limits_{T\in\patset{M}}T$ and $|P| = p > 0$. We need to prove that $p\cdot q \leqslant 2^{k+1}$. Denote upper and lower images $I_b=I^{2^{k+1}}_b, I_t=I^{2^{k+1}}_t$ and the corresponding pattern set restrictions $\patset{M}_b = \{T\cap I_b \mid T\in\patset{M}\}$, $\patset{M}_t = \{T\cap I_t \mid T\in\patset{M}\}$ and consider three cases.

First assume that $P\subseteq I_b$. By proposition~\ref{sibling} at most two bottom siblings in $\patset{M}$ may produce the same element in $\patset{M}_b$,  which yields $q\leqslant 2|\patset{M}_b|$. By induction hypothesis, $S(\patset{M}_b) = p\cdot|\patset{M}_b| \leqslant 2^k$, so $S(\patset{M}) = p \cdot q \leqslant 2 p \cdot |\patset{M}_b| \leqslant 2^{k+1}$.

Secondly, assume that $P\subseteq I_t$. In this case the proof is same as above, we just need to note proposition~\ref{restriction} to use the induction hypothesis.

Finally, assume that both $p_b = |P\cap I_b| > 0$ and $p_t = |P\cap I_t| > 0$. It follows from proposition~\ref{sibling} that in this case $\patset{M}$ contains neither top- nor bottom-siblings and, hence, 
$|\patset{M}_b| = |\patset{M}_t| = q$. By separately applying induction hypothesis to these two sets (and again noting proposition~\ref{restriction}) we get $p_b\cdot q \leqslant 2^k$ and $p_t\cdot q \leqslant 2^k$. For $p = p_b + p_t$ we also get $S(\patset{M}) = p \cdot q \leqslant 2^{k+1}$.

Since $\patset{M}$ has been chosen arbitrarily, we have $r(\patset{H}_{k+1}) \leqslant 2^{k+1}$. The induction step is proven and we get $r(\patset{H}_{h,w}) = r(\patset{H}_d) \leqslant 2^d = h$.

\end{proof}

Remember that we assumed $h=2^d$ for $d\in\mathbb{N}$. Let us further denote $\underline h=2^{\lfloor \log_2 h \rfloor}$ and $\overline h=2^{\lceil \log_2 h \rceil}$, i.e. $\underline h$ and $\overline h$ are the powers of two closest to $h$ from below and above resp. We can extend the definition of $\patset{H}_{h, w}$ for image $I$ with arbitrary $h\in\mathbb{N}$ in the following way:

$$\patset{H}_{h, w} = \{T\cap I \mid T\in \patset{H}_{\overline h, w} \}$$

From proposition~\ref{subimage_prop} follows

\begin{proposition}\label{fht_any_h}
$\mathrm{OR}(\patset{H}_{h,w})\geqslant\mathrm{OR}(\patset{H}_{\underline h, w})$.\qed
\end{proposition}

By combining theorem~\ref{bound_theorem}, lemma~\ref{fht_lemma} and propositions~\ref{fht_size_proposition}, \ref{fht_volume_proposition} and \ref{fht_any_h} we get

\begin{theorem}\label{FHT_OR_theorem}
$\mathrm{OR}(\patset{H}_{h,w}) \geqslant 3 w \underline h \log_3 \underline h$.\qed
\end{theorem}

Remember that the fast Hough transform was originally designed to be calculated in $\Theta(n^2 \log n)$ operations on a $n\times n$ image. In our model it can be strictly formulated with the following proposition. For consistency and because it is easy and does not require much space, we provide its proof from~\cite{brady} adapted to our model.

\begin{proposition}[FHT algorithm]\label{fht_algorithm}
$$\mathrm{SUM2}(\patset{H}_{h,w}) \leqslant w \overline h\log_2 \overline h.$$
\end{proposition}

\begin{proof}
Denoting $d=\lceil \log_2 h \rceil\in\mathbb{N}$ we will inductively construct a fanin-2 circuit that computes $\patset{H}_{\overline h, w}=\patset{H}_d$ with $w d\: 2^d = w \overline h \log_2 \overline h$ gates. For $k=0$ the circuit computing $\patset{H}_0$ consists only of input nodes and thus has zero gates, providing the induction base. Suppose now for $k\geqslant 1$ there is a circuit computing $\patset{H}_k$ with $w k2^k$ gates. Let us replicate it for the top and bottom image halves thus producing $w k2^{k+1}$ gates. Then, following the FHT patterns definition, for every $T\in\patset{H}_k$ we insert two more gates computing its successors $T', T''\in\patset{H}_{k+1}$. For $|\patset{H}_k| = w2^k$, this produces $w2^{k+1}$ more nodes and finishes the circuit. The total number of nodes in it is $wk2^{k+1} + w2^{k+1}=w(k+1)2^{k+1}$. We end the proof by applying proposition~\ref{subimage_prop}.
\end{proof}

The next theorem provides explicit upper and lower bounds on FHT patterns complexity (it looks best for $h=2^d$).

\begin{theorem}\label{FHT_tight_bounds}
For $L\in\{\mathrm{OR}, \mathrm{SUM}\}$,

$$
\log_3 8 \; w \underline h \log_2 \underline h \leqslant L(\patset{H}_{h,w}) \leqslant 2 w \overline h \log_2 \overline h,
$$
and
$$
\log_9 8 \; w \underline h \log_2 \underline h \leqslant L2(\patset{H}_{h,w}) \leqslant w \overline h \log_2 \overline h.
$$
\end{theorem}

\begin{proof}
After applying elementary logarithm properties, the upper right inequality follows from propositions~\ref{sum2_proposition} and \ref{fht_algorithm}, the upper left -- from proposition~\ref{sum_or_ineq} and theorem~\ref{FHT_OR_theorem}, the lower right -- from propositions~\ref{sum_or_ineq} and \ref{fht_algorithm} and the lower left -- from proposition~\ref{sum2_proposition} and theorem~\ref{FHT_OR_theorem}.
\end{proof}

Finally, summarizing and omitting the coefficients we can state

\begin{proposition}
For $L\in\{\mathrm{OR}, \mathrm{SUM}, \mathrm{OR2}, \mathrm{SUM2}\}$,
$L(\patset{H}_{n,n}) = \Theta(n^2 \log n)$.\qed
\end{proposition}

\section{Strip patterns}\label{strip_section}

We will now consider a generic class of strip pattern sets consisting of patterns which have a limited distance to some ``ideal'' straight line and do not form too ``dense'' congestions, the classical Hough transform~(\ref{hough}) being a typical example. As in the previous section we work only with one type of lines (this time mostly horizontal with inclination to the right) which is enough for lower bounds because of proposition~\ref{subpattern_prop}. It is worth noting that FHT patterns from the previous section reach logarithmic deviation from ideal straight lines~\cite{FHT_approx} and thus cannot be reduced to this case.

The following embedding of pixels into $\mathbb{R}^2$ will be used here. Pixel $p_{ij}$ corresponds to integer point $(i,j)\in\mathbb{R}^2$. We use the same $p_{ij}$ notation for these points if it does not lead to confusion.

\begin{definition}
For line $l\subset\mathbb{R}^2$ and $C>0$ $C$-strip or strip of width $C$ is the set
$s(l, C) = \{r\in\mathbb{R}^2 \mid \rho(r, l) \leqslant C/2 \}$. If $C$ is known from the context or not important, this strip is shortly denoted as $s(l)$.
\end{definition}

\begin{definition}
Line $l\subset\mathbb{R}^2$ is mostly horizontal inclined to the right if it has equation of form $y = ax + b$ with slope $0\leqslant a\leqslant1$. In this case strip $s(l)$ is also called mostly horizontal inclined to the right with slope $a$.
\end{definition}

\begin{definition}
Mostly horizontal inclined to the right line $l$ is integer at image $I$ with width $w$ if its slope $a = \frac{e}{w-1}$ for some $e\in\{0, 1, ..., w-1\}$ called elevation.
\end{definition}

\begin{definition}
Patterns $T\in\patset{M}$ are $C$-parallel if there is a family of integer lines $\{l(T) \mid T\in\patset{M}\}$ with the same slope that $\forall T\;T\subset s(l(T), C)$. 
\end{definition}

\noindent We can now describe the specific pattern set class which we study and then obtain the self-intersection bound for it. ``Strangely looking'' item~\ref{strip_patset_def_Q} of the following definition states that there cannot be more than a constant number $Q$ of intersecting patterns covered by strips of the same slope. This is of course true for all standard line plotting algorithms: if $C$ is chosen small and a line is shifted up or down by a few pixels, it would not intersect with itself. The idea behind the $r$-bound proof (lemma~\ref{strip_r_lemma}) is quite simple: if we have a lot of intersecting lines, two of them will have differing enough slopes thus limiting the intersection area and with a few lines an even simpler bound works.

\begin{definition}
Pattern set $\patset{L}$ is a strip pattern set of width $C > 0$ and density $Q\in\mathbb{N}$ if:
\begin{enumerate}
    \item For any pattern $L\in\patset{L}$ there is an integer line $l$ so that $L\subset s(l, C)$,
    \item For any $C$-parallel subset $\patset{M}\subseteq\patset{L}$ with $\bigcap\limits_{L\in\patset{M}}L \ne \varnothing$ holds $|\patset{M}| \leqslant Q$.\label{strip_patset_def_Q}
\end{enumerate}
\end{definition}

\begin{remark}
The ``integer'' property of lines above is not important and is imposed only to shorten the proofs. Indeed, for any mostly horizontal inclined to the right line $l$ there is an integer line $l'$ and constant $C'$ so that $s(l, C)\subseteq s(l', C')$ and $C'$ depends only on $C$.
\end{remark}

\begin{definition}
Integer cardinality of a set $D\subseteq\mathbb{R}^2$ is the number 
$N(D) = |\mathbb{Z}^2 \cap D|$ of integer points inside it.
\end{definition}

\begin{proposition}\label{rect_prop}
Rectangle $R$ with sides $a$ and $b$ has $N(R)\leqslant(a+\sqrt 2)(b+\sqrt 2)$.
\end{proposition}
\begin{proof}
Draw $N(R)$ horizontally aligned squares with unit area centered at the integer points contained in $R$. Obviously, rectangle $R'$ obtained by moving each $R$ side outward by half of the unit square diagonal length ($\frac{\sqrt{2}}{2}$), will contain all the marked squares and hence its area is not less than $N(R)$.
\end{proof}

\begin{lemma}
For mostly horizontal inclined to the right $C$-strips $s(l)$ and $s(k)$ with different slopes $a$ and $b$
$$ N(s(l)\cap s(k)) \leqslant \frac{C_1}{|a-b|}, $$
where $C_1>0$ depends only on $C$.
\end{lemma}

\begin{proof}
Let $\alpha = \arctan a$, $\beta = \arctan b$ and suppose $\beta > \alpha$. The strips sides intersect at angle $\gamma = \beta-\alpha$ and form a rhombus $PQRS$ (fig.~\ref{rhombus}). After dropping perpendiculars $PH_1$ on $QR$ and $RH_2$ on $PS$ we get a rectangle $PH_1RH_2$ with side $PH_1=2C$.

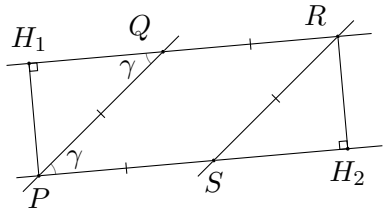
\begin{figure}[h]
\begin{center}

\begin{tikzpicture}
  [
    scale=1.5,
    point/.style = {draw, circle,  fill = black, inner sep = .1pt},
  ]
  
\def \alpha {5}
\def \beta  {45}
\def \w   {1}
\def \side {1.55572}   

\node (P) at (0:0) [point, label = below:$P$] {};
\node (Q) at ($ (P) + (\beta:\side) $) [point, label = above left:$Q$] {};
\node (R) at ($ (Q) + (\alpha:\side) $) [point, label = above left:$R$] { };
\node (S) at ($ (R) - (\beta:\side) $) [point, label = below:$S$]{};

\node (H1) at ($ (P) + (\alpha+90:\w) $) [point, label=above:$H_1$] {};
\node (H2) at ($ (R) + (\alpha+270:\w) $) [point, label=below:$H_2$] {};

\draw (P) -- (Q) -- (R) -- (S) -- (P);
\draw (Q) -- +(\beta:0.2);
\draw (R) -- +(\beta:0.2);
\draw (P) -- +(180+\beta:0.2);
\draw (S) -- +(180+\beta:0.2);
\draw (R) -- +(\alpha:0.3);
\draw (S) -- +(\alpha:1.5);
\draw (Q) -- (H1);
\draw (H1) -- +(180+\alpha:0.2);
\draw (P) -- +(180+\alpha:0.2);

\draw (P) -- (H1);
\draw (R) -- (H2);

\draw[rotate around={\alpha-90:(H1)}] (H1) rectangle +(0.07,0.07);
\draw[rotate around={\alpha+90:(H2)}] (H2) rectangle +(0.07,0.07);

\coordinate (PP) at (P);
\coordinate (SS) at (S);
\coordinate (QQ) at (Q);
\coordinate (HH1) at (H1);

\begin{scope}
\path[clip] (PP) -- (SS) -- (QQ);
\draw[opacity=0.5, draw=black] (P) circle (1.5mm);
\node at ($(P) + (25:3.5mm)$) {$\gamma$};
\end{scope}

\begin{scope}
\path[clip] (QQ) -- (HH1) -- (PP);
\draw[opacity=0.5, draw=black] (Q) circle (1.5mm);
\node at ($(Q) + (207:3.5mm)$) {$\gamma$};
\end{scope}

\coordinate (PQ) at ($(P)!0.5!(Q)$);
\draw[rotate around={90+\beta:(PQ)}]   ($ (PQ) -(0.05,0) $) -- ($ (PQ) +(0.05,0) $);
\coordinate (QR) at ($(Q)!0.5!(R)$);
\draw[rotate around={90+\alpha:(QR)}]   ($ (QR) -(0.05,0) $) -- ($ (QR) +(0.05,0) $);
\coordinate (RS) at ($(R)!0.5!(S)$);
\draw[rotate around={90+\beta:(RS)}]   ($ (RS) -(0.05,0) $) -- ($ (RS) +(0.05,0) $);
\coordinate (SP) at ($(S)!0.5!(P)$);
\draw[rotate around={90+\alpha:(SP)}]   ($ (SP) -(0.05,0) $) -- ($ (SP) +(0.05,0) $);

\end{tikzpicture}
\end{center}
\caption{Strips intersection.} \label{rhombus}
\end{figure}

Another side $H_1R = H_1Q+QP = 2C(\frac{1}{\tan\gamma} + \frac{1}{\sin\gamma})$. After denoting $z = \tan\gamma \leqslant 1$ and applying elementary trigonometric equalities we have $H_1R = 2C(\frac1z + \frac{\sqrt{z^2+1}}{z}) \leqslant \frac{6C}{z} = 6C\frac{1+ab}{b-a} \leqslant \frac{12C}{b-a}$ for all $a, b, z \leqslant 1$.

By proposition~\ref{rect_prop} $\;\; N(s(l)\cap s(k)) \leqslant N(PH_1RH_2) \leqslant (PH_1 + \sqrt{2})(H_1R + \sqrt{2}) \leqslant (2C+\sqrt{2})(\frac{12C}{b-a} + \sqrt{2})$.
\end{proof}

\begin{corollary}\label{intersect_count}
For integer lines $l, k$ with different elevations $e_1$ and $e_2$
$$ N(s(l, C)\cap s(k, C)) \leqslant \frac{C_1 w}{|e_1-e_2|}, $$
where $C_1>0$ depends only on $C$, $w$ is image width.\qed
\end{corollary}

Further we will utilize the following obvious
\begin{proposition}\label{strip_count}
For mostly horizontal inclined to the right line $l$ and image $I$ of width $w$, $N(I\cap s(l,C)) \leqslant (2C\sqrt 2+1) w.$\qed
\end{proposition}

\begin{lemma}\label{strip_r_lemma}
For any strip pattern set $\patset{L}$ of width $C$ and density $Q$ at image $I$ with width $w$, $$r(\patset{L})\leqslant C' w,$$ where $C'>0$ depends only on $C$ and $Q$.
\end{lemma}

\begin{proof}
Suppose we have patterns subset $\patset{M} = \{L_i\}_{i=1}^m\subseteq\patset{L}$ such that $|\bigcap\limits_{i=1}^m L_i| = p > 0$. By strip pattern definition, we have a set of integer lines $\{l_i\}$ so that each $L_i\subset s(l_i)$. Consider the set $\{e_j\}_{j=1}^k$ of different $l_i$ elevations, $1\leqslant k\leqslant m$ and denote by $\patset{M}_j\subseteq\patset{M}$ the C-parallel subset of patterns covered by lines $l_i$ with the same slope $e_j$, $\patset{M} = \bigsqcup
\limits_{j=1}^k\patset{M}_j$ Again, by strip pattern definition, $|\patset{M}_j|\leqslant Q$ for all $j$, which yields $k\geqslant m/Q$.

First suppose that $m \geqslant 2Q$, so $k\geqslant 2$. In this case we can find two distinct indices $j_1$ and  $j_2$ so that $e_{j_2}-e_{j_1}\geqslant k-1\geqslant m/Q-1$. By corollary~\ref{intersect_count},
$$p \leqslant N(s(l_{j_1}), s(l_{j_2})) \leqslant \frac{C_1 w}{e_{j_2}-e_{j_1}}\leqslant \frac{C_1 w}{m/Q-1},$$
hence
$$S(\patset{M}) = pm \leqslant C_1 Q \,w \,\frac{m}{m-Q}\leqslant 2C_1 Q w.$$

Suppose now that $m<2Q$. In this case by proposition~{\ref{strip_count}}
$$S(\patset{M}) \leqslant 2Q (2C\sqrt 2+1) w.$$

\noindent By setting $C' = \max(2C_1Q, 2Q(2C\sqrt 2+1))$ and noting that $\patset{M}$ is arbitrary we finish the proof.
\end{proof}

From theorem~\ref{bound_theorem}, lemma~\ref{strip_r_lemma} and propositions~\ref{sum2_proposition}, \ref{sum_or_ineq} follows

\begin{theorem}\label{Ln_theorem}
If $L\in\{\mathrm{OR}, \mathrm{SUM}, \mathrm{OR2}, \mathrm{SUM2}\}$,  $\{\patset{L}_n\}$ is a family of strip pattern sets of the same width and density, each $\patset{L}_n$ defined at image $I_n$ of size $n\times n$ and for all $\patset{L}_n$:
\begin{enumerate}
    \item $|\patset{L}_n| = \Theta(n^2)$,
    \item $V(\patset{L}_n) = \Omega(n^3)$,\label{Ln_theorem_large}
\end{enumerate}
then $\quad L(\patset{L}_n) = \Omega(n^2 \log n)$.\qed
\end{theorem}

The conditions of the theorem include all pattern sets which are enough ``large'' (item~\ref{Ln_theorem_large}) and have mean pattern size of $\Omega(n)$. Any lines parametrization from~(\ref{hough}) which contains $\Theta(n)$ different slopes and $\Theta(n)$ lines for each slope falls under these conditions.

\section{Conclusion}\label{conclusion}
In this paper we successfully applied a certain result from the theory of boolean circuits to obtain the $\Omega(n^2 \log n)$ asymptotic lower bounds on additive complexities of the fast Hough transform and strip pattern based discrete Radon transform. For the fast Hough transform this bound is $\Theta$-exact.

In fact a stronger result was proved, limiting from below the so called OR-complexity. Also, for the fast Hough transform pairs of simple boundary inequalities with explicit coefficient were provided, covering all the SUM, OR, SUM2 and OR2 cases.

We can outline a few further research directions. An intriguing question is if strip patterns complexity is $o(n^3)$ and, if so, whether the corresponding circuits have constructive description. Obtaining better lower bounds than the proven $\Omega(n^2\log n)$ would also be interesting. Another question is whether the FHT and strip patterns asymptotic complexities would change after allowing to use the subtraction operation.

It seems that the described approach could be applied to analyze the complexity of the generalized Hough transform patterns. Simple shapes might allow something similar to the geometrical analysis of $r(\patset{T})$ which we used in this paper.

\section*{Acknowledgements}
I would like to thank my colleagues Dmitry Nikolaev for bringing this subject to my attention and useful advice, Valerii Sokolov for numerous remarks and correcting certain inconsistencies, Igor Polyakov and Alexey Savchik for verifying the proofs and locating a few misprints, and Andrey Gladkov for help with the FHT pictures.

\end{document}